\documentclass[journal,doublecolum,10pt]{IEEEtran}

\usepackage{epsfig,latexsym}
\usepackage{float}
\usepackage{indentfirst}
\usepackage{amsmath}
\usepackage{amssymb}
\usepackage{times}
\usepackage{subfigure}
\usepackage{psfrag}
\usepackage{cite}
\usepackage{lastpage}
\usepackage{fancyhdr}
\usepackage{color}
 \usepackage{amsthm}
\usepackage{bigints}
\newtheorem{Proposition}{Proposition}
\newtheorem{Lemma}{Lemma}
\newtheorem{lemma}[Lemma]{$\mathbf{Lemma}$}
\newtheorem{proposition}[Proposition]{Proposition}
\begin{document}%
\title{ {\huge Cooperative Non-Orthogonal Multiple Access in 5G Systems }}

\author{ Zhiguo Ding, \IEEEmembership{Member, IEEE}, Mugen Peng, \IEEEmembership{Senior Member, IEEE},
  and  H. Vincent Poor, \IEEEmembership{Fellow, IEEE}\thanks{
The authors  are with the Department of
Electrical Engineering, Princeton University, Princeton, NJ 08544,
USA.   Z. Ding is also with the School of
Computing and Communications, Lancaster
University, LA1 4WA, UK. M. Peng is also with the Institute of Telecommunications, Beijing University of Posts and Telecommunications, Beijing, China.  }\vspace{-0.5em}} \maketitle
\begin{abstract}
 Non-orthogonal multiple access (NOMA) has recently received considerable attention as a promising candidate for 5G systems.  A key feature of NOMA is that users with better channel conditions have   prior information about the messages of the other users. This prior knowledge is fully exploited in this paper, where a  cooperative NOMA scheme is proposed. Outage probability and diversity order achieved by this cooperative NOMA scheme are analyzed, and an approach based on user pairing is also proposed to reduce system complexity in practice.
\end{abstract}\vspace{-1.5em}
 \section{Introduction}
 Non-orthogonal multiple access (NOMA) is  fundamentally different from conventional orthogonal multiple access (MA) schemes, as in NOMA multiple users are encouraged to transmit at the same time, code and frequency, but with different power levels \cite{NOMAPIMRC}. In particular, NOMA allocates less power to the users with better channel conditions, and these users can decode their own information by applying successive interference cancellation \cite{Cover1991}. Consequently the users with better channel conditions will know the messages intended  to the others; however, such prior information has not been  exploited by the existing works about NOMA \cite{6708131} and \cite{Nomading}.

 In this paper, a cooperative NOMA transmission scheme is proposed by fully  exploiting prior information available in NOMA systems. In particular, the use of the successive detection strategy at the receivers  means that  users with better channel conditions need to decode the messages for the others, and therefore these users can be used as relays to improve the reception reliability for the users with poor connections to the base station. Local short-range communication techniques, such as bluetooth and ultra-wideband (UWB), can be used to deliver messages from the users with better channel conditions to the ones with  poor channel conditions.    The outage probability and diversity order achieved by this cooperative NOMA scheme are analyzed, and these analytical results demonstrate that cooperative NOMA can achieve the maximum diversity gain for all the users. In practice, inviting all users in the network   to participate in cooperative NOMA might not be realistic  due to   two reasons. One is that a large amount of system overhead will be consumed to coordinate multi-user networks, and the other is that user cooperation will   consume extra short-range communication resources. User pairing is a promising solution to reduce system complexity, and we demonstrate that grouping   users with high channel quality does not necessarily yield a large performance gain over orthogonal MA. Instead, it is   preferable to pair   users whose channel gains, the absolute squares of the channel coefficients,  are more distinctive.

\section{System Model}
Consider a broadcast channel with one base station (the source), and $K$ users (the destinations). Cooperative NOMA consists of two phases, as described in the following.\vspace{-1em}
\subsection{Direct Transmission Phase}
During this phase, the base station sends $K$ messages to the destinations based on the NOMA principle, i.e., the base station sends $\sum^{K}_{m=1}p_ms_m$, where $s_m$ is the message for the $m$-th user, and $p_m$ is the power allocation coefficient. The observation at the  $k$-th user is given by
\begin{eqnarray}
 {y}_{1,k} &=& \sum^{K}_{m=1} {h}_k  {p}_m {s}_m+ {n}_k,
\end{eqnarray}
where $h_k$ denotes the Rayleigh fading channel coefficient from the base station to the $k$-th user and $n_k$ denote the additive Gaussian noise. Without loss of generality, consider that the users are ordered based on their channel quality, i.e.,
\begin{eqnarray}\label{order}
|h_1|^2\leq \cdots \leq |h_K|^2.
\end{eqnarray}
The use of NOMA implies   $|p_1|^2\geq \cdots\geq |p_K|^2$, with $\sum^{K}_{m=1}   {p}_m^2=1$.
Successive detection will be carried out at the $K$-th user at the end of this phase. The receiving signal to noise ratio (SNR) for  the $K$-th   ordered user to detect the $k$-th user's message, $1\leq k<K$,  is given by
\begin{eqnarray}\label{bound capacity}
SNR_{K,k} &=&   \frac{| {h}_K|^2  |{p}_k|^2}{\sum^{K}_{m=k+1}| {h}_K^H  {p}_m|^2 +\frac{1}{\rho}},
\end{eqnarray}
where $\rho$ is the transmit SNR. After these users' messages are decoded, the $K$-th   user can decode its own information at the following SNR
\begin{eqnarray}
SNR_{K,K} &=&    \rho| {h}_K|^2  |{p}_K|^2.
\end{eqnarray}
Therefore the conditions under which  the $K$-th user can decode its own information  are given by $
\log (1+SNR_{K,k})>R_k,\quad \forall 1\leq k\leq K$,
  where $R_k$ denotes the targeted data rate for the $k$-th user.

\vspace{-1em}
\subsection{Cooperative Phase}
During this phase, the users cooperate  with each other via short range communication channels. Particularly the second phase consists of $(K-1)$ time slots. During the first time slot, the $K$-th user broadcasts the combination of the $(K-1)$ messages with the coefficients $\mathbf{q}_K$, i.e., $\sum^{K-1}_{m=1}q_{K,m}s_m$ and $\sum^{K-1}_{m=1}q_{K,m}^2=1$, where $\sum^{K-1}_{m=1}q_{K,m}^2=1$. The $k$-th user observes the following
\begin{eqnarray}
 {y}_{2,k} &=&   \sum^{K-1}_{m=1} {g}_{K,k}  {q}_{K,m} {s}_m+ {n}_{2,k},
\end{eqnarray}
for $k<K$, where $g_{K,k}$ denotes the inter-user channel coefficient. The $(K-1)$-th user uses maximum ratio combining to combine the observations from both phases, and the SNR for this user to decode the $k$-th user's message, $k<(K-1)$, is  given by
\begin{eqnarray}
SNR_{K-1,k} &=& \frac{|h_{K-1}|^2p_k^2}{|h_{K-1}|^2\sum^{K}_{m=k+1}p_{m}^2
+\frac{1}{\rho}}\\ \nonumber &&+\frac{|g_{K,K-1}|^2q_{K,k}^2}{|g_{K,K-1}|^2\sum^{K-1}_{m=k+1}q_{K,m}^2+\frac{1}{\rho}}.
\end{eqnarray}
After the $(K-1)$-th user decodes the other users' messages, it can decode its own information with the following SNR
\begin{eqnarray}
SNR_{K-1,K-1} = \frac{|h_{K-1}|^2p_{K-1}^2}{|h_{K-1}|^2p_{K}^2
+\frac{1}{\rho}}+ |g_{K,K-1}|^2q_{K,K-1}^2.
\end{eqnarray}
Similarly at the $n$-th time slot, $1\leq n \leq (K-1)$,    the $(K-n+1)$-th user broadcasts the combination of the $(K-n)$ messages with the coefficients ${q}_{K-n+1,m}$, i.e., $\sum^{K-n}_{m=1}q_{K-n+1,m}s_m$. The $k$-th user, $k<(K-n+1)$,  observes
\begin{eqnarray}
 {y}_{2,k} &=&   \sum^{K-n}_{m=1} {g}_{K-n+1,k}^H  {q}_{K-n+1,m} {s}_m+ {n}_{n+1,k}.
\end{eqnarray}
 Combining the observations from  both phases, the $(K-n)$-th user can decode the $k$-th user's message, $1\leq k<(K-n)$, with the following SNR
\begin{align}\label{eq1}
&SNR_{K-n,k} = \frac{|h_{K-n}|^2p_k^2}{|h_{K-n}|^2\sum^{K}_{m=k+1}p_{m}^2
+\frac{1}{\rho}}\\ \nonumber & +\sum^{n}_{i=1}\frac{|g_{K-i+1,K-n}|^2q_{K-i+1,k}^2}{|g_{K-i+1,K-n}|^2\sum^{K-i}_{m=k+1}q_{K-i+1,m}^2+\frac{1}{\rho}},
\end{align}
and it can decode its own information with the following SNR
\begin{align}\label{eq2}
&SNR_{K-n,K-n} = \frac{|h_{K-n}|^2p_{K-n}^2}{|h_{K-n}|^2\sum^{K}_{m=K-n+1}p_{m}^2
+\frac{1}{\rho}} \\\nonumber & +\sum^{n-1}_{i=1}\frac{|g_{K-i+1,K-n}|^2q_{K-i+1,K-n}^2}{|g_{K-i+1,K-n}|^2\sum^{K-i}_{m=K-n+1}q_{K-i+1,m}^2+\frac{1}{\rho}}\\\nonumber &+ \rho|g_{K-n+1,K-n}|^2q_{K-n+1,K-n}^2.
\end{align}
Recall that, without cooperation, the SNR at the $k$-th user is $ \frac{|h_{K-n}|^2p_{K-n}^2}{|h_{K-n}|^2\sum^{K}_{m=K-n+1}p_{m}^2+\frac{1}{\rho}} $. Compared it to \eqref{eq2}, one can find out that the use of cooperation can boost reception reliability.

\vspace{-1em}

\section{Performance Analysis}
Provided  that the $(n-1)$ best users can achieve reliable detection, the outage probability for the $(K-n)$-th user can be expressed as follows:\footnote{Because of the use of short-range communications, the cooperative phase does not consume any cellular frequency, i.e.,  $\epsilon_k= 2^{R_k}-1$. Without using short-range communications, the targeted receive SNR becomes  $\epsilon_k= 2^{\frac{R_k}{K}}-1$, but the analytical results about the diversity order obtained in this paper are still valid with some straightforward modifications.}
\begin{eqnarray}
\mathrm{P}_o^{K-n} \triangleq \mathrm{P}(SNR_{K-n,k}<\epsilon_k, \forall k\in\{1, \cdots, K-n\}),
\end{eqnarray}
where $\epsilon_k= 2^{R_k}-1 $. Note that the use of local short-range communications does not reduce the data rate.  For notational simplicity, define $a_{k,i}^{K-n} = q_{K-i+1,k}^2$ and $b_{k,i}^{K-n} =  \sum^{K-i}_{m=k+1}q_{K-i+1,m}^2 $,  where $1\leq k\leq (K-n)$ and $1\leq i \leq n$ with the special case of $a_{K-n,n}^{K-n} =   q_{K-n+1,K-n}^2 $ and $b_{K-n,n}^{K-n} = 0 $. In addition, define $a_{k,0}^{K-n}=p_k^2$ and $b_{k,0}^{K-n}= \sum^{K}_{m=k+1}p_{m}^2$,  for $1\leq k \leq (K-n)$.   By using the   definition of the outage probability, we can have the following proposition  for the diversity order achieved by the proposed cooperative NOMA scheme.
\begin{proposition}\label{propositon}
Assume  that the $(n-1)$ best users can achieve reliable detection. The proposed cooperative NOMA scheme can ensure that the $(K-n)$-th ordered user  experiences a diversity order of $K$, conditioned on   $\epsilon_{k}< \frac{a_{k,i}^{K-n}}{b_{k,i}^{K-n}}$, for $1\leq k \leq (K-n)$ and $0\leq i \leq n$.
\end{proposition}
\begin{proof}
For notational simplicity, define $z_{k,i}^{K-n} =\frac{|g_{K-i+1,K-n}|^2q_{K-i+1,k}^2}{|g_{K-i+1,K-n}|^2\sum^{K-i}_{m=k+1}q_{K-i+1,m}^2+\frac{1}{\rho}}$, where $1\leq k\leq(K-n)$ and $1\leq i \leq n$, except $z^{K-n}_{K-n,n}=\rho|g_{K-n+1,K-n}|^2q_{K-n+1,K-n}^2$. In addition, define $z_{k,0}^{K-n}=\frac{|h_{K-n}|^2p_k^2}{|h_{K-n}|^2\sum^{K}_{m=k+1}p_{m}^2+\frac{1}{\rho}}$.  The SNRs can be expressed as follows:
\begin{eqnarray}\label{eq3}
SNR_{K-n,k} = z_{k,0}^{K-n}
 +\sum^{n}_{i=1}  z_{k,i}^{K-n},
\end{eqnarray}
for $1\leq k \leq (K-n)$. Therefore the outage probability can be rewritten as follows:
\begin{align} \nonumber
\mathrm{P}_o^{K-n} &= \mathrm{P}\left( z_{k,0}^{K-n}
 +\sum^{n}_{i=1}  z_{k,i}^{K-n}<\epsilon_k, \forall k\in\{1, \cdots, K-n\}\right)\\   &\leq
 \sum^{K-n}_{k=1}\mathrm{P}\left( z_{k,0}^{K-n}
 +\sum^{n}_{i=1}  z_{k,i}^{K-n}<\epsilon_k\right),
\end{align}
since $\mathrm{P}(A\cup B)\leq \mathrm{P}(A)+\mathrm{P}(B)$.
Because channel gains are independent and  $\mathrm{P}(a+b<c)\leq \mathrm{P}(a<c)+\mathrm{P}(b<c)$, the outage probability can be further bounded as follows:
\begin{eqnarray}\label{eq44}
\mathrm{P}_o^{K-n}   &\leq& \sum^{K-n}_{k=1}\prod^{n}_{i=0}\mathrm{P}\left( z_{k,i}^{K-n} <\epsilon_k\right),
\end{eqnarray}

All the  elements in \eqref{eq3} except $z_{k,0}^{K-n}$ and $z^{K-n}_{K-n,n}$   share the same structure as follows:
\begin{eqnarray}\label{eqxx1}
z_{k,i}^{K-n}=\frac{a_{k,i}^{K-n}x}{b_{k,i}^{K-n}x+\frac{1}{\rho}}.
\end{eqnarray}
When $x$ is exponentially distributed, the cumulative density function (CDF) of $z_{k,i}^{K-n}$ is given by{\small
\begin{eqnarray}
\mathrm{P}_{z_{k,i}^{K-n}}(Z<z)= \left\{\begin{array}{ll} 1,&  \hspace{-1em}if\quad  \hspace{-0.4em}  z\geq \frac{a_{k,i}^{K-n}}{b_{k,i}^{K-n}} \\
1-e^{-\frac{z}{\rho\left(a_{k,i}^{K-n}-b_{k,i}^{K-n}z\right)}}, & otherwise  \end{array}\right . \hspace{-1em},
\end{eqnarray}}
where the definitions for  the coefficients $a_{k,i}^{K-n}$ and $b_{k,i}^{K-n}$ are given in the proposition.

 At high SNR,  $\frac{\epsilon_k}{\rho\left(a_{k,i}^{K-n}-b_{k,i}^{K-n}z\right)}\rightarrow 0$, and the probability for the  event, $  z_{k,i}^{K-n} <\epsilon_k $, can be approximated by using the power series of exponential functions \cite{GRADSHTEYN} as follows:
\begin{eqnarray}\label{eq43}
\mathrm{P}_{z_{k,i}^{K-n}}\left(Z< \epsilon_k\right) = 1-e^{-\frac{\epsilon_k}{\rho\left(a_{k,i}^{K-n}-b_{k,i}^{K-n}\epsilon_k\right)}}
  \approx  \frac{\epsilon_k }{\rho a_{k,i}^{K-n} },
\end{eqnarray}
which is conditioned on  $ \epsilon_k< \frac{a_{k,i}^{K-n}}{b_{k,i}^{K-n}} $.

 The density functions of the two special cases, $z_{k,0}^{K-n}$ and $z^{K-n}_{K-n,n}$, can be obtained as follows. Note that the source-user channels are sorted according to their quality. By applying order statistics \cite{David03}, the CDF of $z_{k,0}^{K-n}$ can be found as follows:\vspace{-0.5em}
 \begin{align}
&\mathrm{P}_{z_{k,0}^{K-n}}(Z<z)=\\\nonumber
 &\left\{\begin{array}{ll} 1, & \hspace{-0.5em} if\quad  z\geq \frac{a_{k,0}^{K-n}}{b_{k,0}^{K-n}} \\
\bigintss^{\frac{\frac{z}{\rho}}{a_{k,0}^{K-n}-b_{k,0}^{K-n}z}}_{\hspace{-3em} 0\hspace{2.5em}}\frac{e^{-x}}{(K-n-1)!}x^{K-n-1}dx, & otherwise  \end{array}\right . .
 \end{align}

Again applying the high SNR approximation,  the  probability, $ \mathrm{P}(z_{k,0}^{K-n} <\epsilon_k)$, can be approximated by using the power series of exponential functions \cite{GRADSHTEYN} as follows:
\begin{align}\nonumber
\mathrm{P}\left( z_{k,0}^{K-n} <\epsilon_k\right) &= \int^{\frac{\epsilon_k}{\rho\left(a_{k,i}^{K-n}-b_{k,i}^{K-n}\epsilon_k\right)}}_{ 0 }\frac{x^{K-n-1}e^{-x}}{(K-n-1)!}dx
\\ \label{eq42} &\approx \frac{\epsilon_k^{K-n}}{(K-n)!\left(a_{k,i}^{K-n}\right)^{K-n}\rho^{K-n}},
\end{align}
conditioned on $\epsilon_{k}< \frac{a_{k,0}^{K-n}}{b_{k,0}^{K-n}}$. Similarly the probability for the event $z^{K-n}_{K-n,n} <\epsilon_k$ can be approximated  as follows:
\begin{eqnarray}\label{eq41}
\mathrm{P}(z^{K-n}_{K-n,n}<\epsilon_{k}) \approx \frac{\epsilon_{k}}{  q_{K-n+1,K-n}^2 \rho},
\end{eqnarray}
since $z^{K-n}_{K-n,n}$ can be treated as a special case of  \eqref{eqxx1}.

Combining \eqref{eq44}, \eqref{eq43}, \eqref{eq42} and \eqref{eq41}, the  diversity order achieved by the cooperative NOMA scheme can be obtained, which completes the proof.
\end{proof}
The overall system outage event is defined as the event that any user in the system cannot achieve reliable detection, which means the overall outage probability is defined as follows:
\begin{eqnarray}
\mathrm{P}_o \triangleq 1- \prod^{K}_{k=1} \left(1-\mathrm{P}_o^{k}\right).
\end{eqnarray}
By using Proposition \ref{propositon} and the fact that the source-destination channels are independent, the following lemma can be  obtained straightforwardly.
\begin{lemma}\label{lemma}
The proposed cooperative NOMA scheme can ensure that the $n$-th best user, $1\leq n \leq K$,  experiences a diversity order of $K$, conditioned on $\epsilon_{k}< \frac{a_{k,i}^{K-n}}{b_{k,i}^{K-n}}$, for $1\leq k \leq (K-n)$ and $0\leq i \leq n$.
\end{lemma}
 This diversity order result is not surprising as explained in the following. Take  the user with the worst channel connection to the source  as an example. When cooperative NOMA is implemented, it gets help from the other $(K-1)$ users, in addition to its own direct channel to the source, which implies that  the number of independent paths from the source to this user is $K$, i.e., the achievable  diversity order for this user is $K$.    In general,   cooperative NOMA can efficiently exploit  user cooperation and  ensure that   a diversity order of $K$  is achievable by all users, regardless of their channel conditions,  whereas  non-cooperative NOMA can achieve only a diversity order of $n$ for the $n$-th ordered user \cite{Nomading}.
\subsection*{Reducing System Complexity via User Pairing}\label{subsection pair}\vspace{-0.5em}
Practical implementation of cooperative NOMA may face some challenges, such as large time delay,     extra system overhead  for coordinating multiple users, as well as  additional short-range communication bandwidth resources consumed for cooperation. This motivates the study of user pairing/grouping. Particularly it is more practical to divide the users in one cell into multiple groups, where cooperative NOMA is implemented within each  group and conventional MA can be used for inter-group multiple access. Since there are fewer users in each group to participate in cooperative NOMA in this hybrid MA system, the aforementioned challenges    can be effectively mitigated.   Without loss of generality,  we focus on the case to select only two users. An important question to be answered here is which two users should be grouped together.

 Consider that the users are ordered as \eqref{order}, and the  $m$-th and $n$-th users are paired together, $m<n$. The conventional TDMA can achieve the following rates
\begin{eqnarray}
 \bar{R}_{m}=\frac{1}{2}\log\left(1+\rho |h_{m}|^2\right), \quad \bar{R}_{n}=\frac{1}{2}\log\left(1+\rho |h_{n}|^2\right).
\end{eqnarray}

The rates achieved by cooperative NOMA is quite complicated, so we first consider conventional NOMA which can achieve the following rates
\begin{eqnarray}
R_{m} &=& \log\left(1+  \frac{\rho| h_{m}|^2  {p}_{m}^2}{\rho |{h}_{m}|^2  {p}_{n}^2  +1}\right),
\end{eqnarray}and $
R_{n} =  \log\left(1+\rho p_{n}^2 |h_{n}|^2\right)$,
where $R_{n}$ is achievable since $\log\left(1+\frac{| {h}_{n}|^2  {p}_{m}^2}{ |{h}_{n}|^2  {p}_{n}^2  +1}\right)\geq R_m$.

The gap between the two sum rates achieved by TDMA and conventional  NOMA can be expressed as follows:
\begin{align}\label{diff}
&R_m+R_n-\bar{R}_m-\bar{R}_n
\\\nonumber \approx & \log\left(1+  \frac{   {p}_{m}^2}{   {p}_{n}^2  }\right)+\log \rho p_n^2|h_n|^2 -\frac{\log \rho |h_{m}|^2}{2}-\frac{\log \rho |h_{n}|^2}{2}
\\\nonumber = &\frac{\log   |h_{n}|^2}{2}-\frac{\log   |h_{m}|^2}{2},
\end{align}
where the approximation is obtained at high SNR.
It is interesting to observe that the gap is not a function of power allocation coefficients $p_m$, but depends on how different the two users' channels are.
Therefore to conventional NOMA, the worst choice of $m$ and $n$ is $n=m+1$, and it is ideal to group two users who experience significantly different channel fading. This observation is also valid to cooperative NOMA. Particularly an important observation from \eqref{bound capacity} is that the data rate for the $m$-th user is bounded as $
R_{m} \leq \log\left(1+  \frac{\rho| h_{n}|^2  {p}_{m}^2}{\rho |{h}_{n}|^2  {p}_{n}^2  +1}\right)$,
although $R_m$ can be as large as $ \log\left(1+  \frac{\rho| h_{m}|^2  {p}_{m}^2}{\rho |{h}_{m}|^2  {p}_{n}^2  +1}+\rho|g_{n,m}|^2\right)$, where the bound is due to the fact that the $n$-th user needs to decode the $m$-th user's information.  Since $\log\left(1+  \frac{\rho| h_{n}|^2  {p}_{m}^2}{\rho |{h}_{n}|^2  {p}_{n}^2  +1}\right)\approx \log\left(1+  \frac{   {p}_{m}^2}{   {p}_{n}^2  }\right)$, the conclusion obtained for conventional NOMA can also be applied  to cooperative NOMA.

\begin{figure}[!htbp]\centering\vspace{-1.5em}
    \epsfig{file=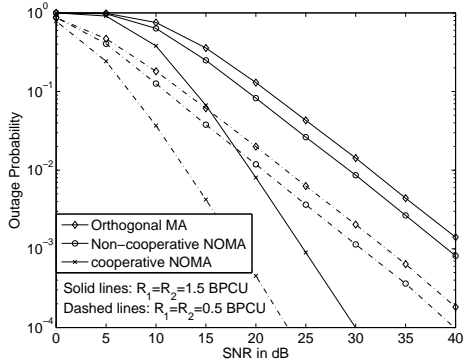, width=0.35\textwidth, clip=}\vspace{-1em}
\caption{Outage probability achieved by cooperative NOMA.}\label{fig1}\vspace{-2em}
\end{figure}

\section{Numerical Studies}\label{section numeral}
In this section, the performance of cooperative NOMA is evaluated  by using computer simulations.  In Fig. \ref{fig1}, the outage probability achieved by the three schemes, e.g., the orthogonal MA scheme,  non-cooperative NOMA, and cooperative NOMA, is shown as a function of SNR, with $K=2$ and $p_1^2=\frac{4}{5}$. As can be seen from the figure,   cooperative NOMA   outperforms the other two schemes, since it can ensure that the maximum diversity gain is achievable to all the users  as indicated by Lemma \ref{lemma}.
In Fig. \ref{fig2}, the outage capacity achieved by the three schemes is demonstrated, by setting  $R_1=R_2$. With $10\%$ outage probability and the transmit SNR equal to $15$ dB, the orthogonal MA scheme can achieve a rate of   $0.7$ bits per channel use (BPCU),  non-cooperative NOMA can support $0.95$ BPCU, and cooperative NOMA can support $1.7$ BPCU, much larger than the other two schemes.   Fig. \ref{fig31}   demonstrates that  the proposed cooperative NOMA scheme can still outperform the comparable schemes, particularly at high SNR, even if   local short-range communication bandwidth resources are not available. Note that without using short range communications, extra $(M-1)$ time slots are used for cooperation transmissions.

 In Fig. \ref{fig3}, the impact of user pairing is investigated by studying the difference between the sum rates achieved by the orthogonal MA scheme and NOMA. Particularly, suppose  that  the $K$-th ordered user, i.e., the user with the best channel condition, is scheduled, and Fig. \ref{fig3} demonstrates how large a sum rate gain can be obtained by pairing it with different  users. As discussed in Section \ref{subsection pair},   without careful user scheduling, the benefit of using NOMA is diminishing.  Such a conclusion is confirmed by the results shown in Fig. \ref{fig3}, where pairing the $K$-th user with  the first user, i.e., the user with the worst channel condition, can yield a significant gain. This observation is also  consistent to the motivation of NOMA in \cite{NOMAPIMRC} which is to schedule two users, one close to the cell edge and the other close to the source.\vspace{-1em}
\begin{figure}[htbp]\centering
    \epsfig{file=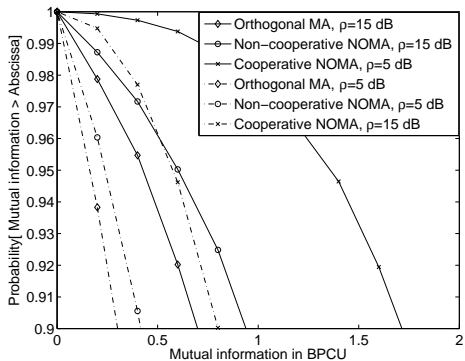, width=0.34\textwidth,clip=}\vspace{-1em}
    \caption{Outage capacity achieved by cooperative NOMA.}\label{fig2}\vspace{-1.5em}
\end{figure}
\vspace{-0.5em}\section{Conclusions}
In this paper, we have proposed a cooperative NOMA transmission scheme which fully uses the fact that some users in NOMA systems have prior information about the others' messages. Analytical results have been developed to demonstrate the performance gain of this cooperative NOMA scheme.     It has been recognized that optimizing power allocation coefficients can improve the performance of non-cooperative NOMA \cite{6692251, 6735800} and it is a promising future direction to study optimal power allocation in cooperative NOMA systems for further performance improvement.

\begin{figure}[htbp]\centering \vspace{-1em}
    \epsfig{file=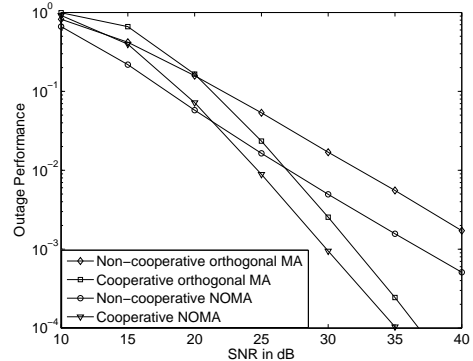, width=0.34\textwidth,clip=}\vspace{-1em}
    \caption{Outage probability achieved by cooperative NOMA without using local short-range communications. $R_1=1.2$ BPCU and $R_2=1.9$ BPCU. }\label{fig31}\vspace{-1em}
\end{figure}

\begin{figure}[htbp]\centering \vspace{-1em}
    \epsfig{file=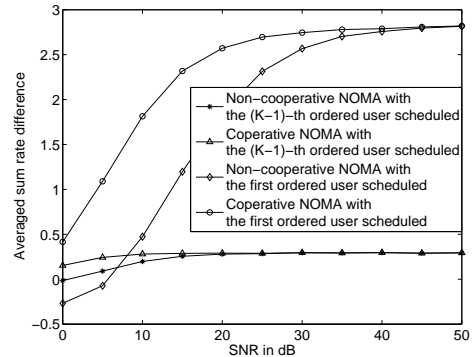, width=0.34\textwidth,clip=}\vspace{-1em}
    \caption{The impact of user pairing on the sum rate. $K=10$.  }\label{fig3}\vspace{-1.5em}
\end{figure}

 \bibliographystyle{IEEEtran}
\bibliography{IEEEfull,trasfer}

  \end{document}